\documentclass[10pt,envcountsect]{llncs}

\usepackage{amsmath,hyperref}
\hypersetup{
    colorlinks,%
    citecolor=blue,%
    filecolor=blue,%
    linkcolor=blue,%
    urlcolor=blue
}
\usepackage{amssymb,verbatim}
\usepackage{fix2col,listings,fancyvrb}
\usepackage{multicol}
\usepackage{graphicx,epsfig}
\usepackage{caption}
\usepackage{stmaryrd}
\usepackage{setspace}
\usepackage{ulem}
\usepackage{newlfont}
\usepackage{epsfig,graphics}
\usepackage{fancybox}
\usepackage{listings}
\usepackage{caption}
\usepackage{graphicx,epsfig}
\usepackage{caption}
\usepackage{makeidx}
\makeindex

\usepackage{hyperref}
\hypersetup{
    colorlinks,%
    citecolor=blue,%
    filecolor=blue,%
    linkcolor=blue,%
    urlcolor=blue
}

\newcommand{\flow}{\mathsf{flow}}
\newcommand{\jump}{\mathsf{jump}}
\newcommand{\inv}{\mathsf{inv}}
\newcommand{\init}{\mathsf{init}}

\newcommand{\reach}{\mathsf{Reach}}
\newcommand{\unsafe}{\mathsf{unsafe}}

\newcommand{\np}{\mathsf{NP}}

\newcommand{\lrf}{\mathcal{L}_{\mathbb{R}_{\mathcal{F}}}}
\newtheorem{notation}[theorem]{Notation}
\newcommand{\enforce}{\mathsf{enforce}}
\newcommand{\traj}{\mathsf{traj}}
\title{$\delta$-Complete Analysis for Bounded Reachability of Hybrid Systems}
\author{Sicun Gao \and Soonho Kong \and Wei Chen \and Edmund M. Clarke}
\institute{Carnegie Mellon University, Pittsburgh, PA 15213}

\begin{document}
\maketitle

\begin{abstract}
We present the framework of $\delta$-complete analysis for bounded reachability problems of general hybrid systems. We perform bounded reachability checking through solving $\delta$-decision problems over the reals. The techniques take into account of robustness properties of the systems under numerical perturbations. We prove that the verification problems become much more mathematically tractable in this new framework. Our implementation of the techniques, an open-source tool {\sf dReach}, scales well on several highly nonlinear hybrid system models that arise in biomedical and robotics applications.
\end{abstract}

\section{Introduction}
Formal verification is difficult for hybrid systems with nonlinear dynamics and complex discrete controls~\cite{DBLP:conf/emsoft/Alur11,DBLP:conf/lics/Henzinger96}. A major difficulty of applying advanced verification techniques in this domain comes from the need of solving logic formulas over the real numbers with nonlinear functions, which is notoriously hard. Recently, we have defined the {\em $\delta$-decision problem} that is much easier to solve~\cite{DBLP:conf/lics/GaoAC12,DBLP:conf/cade/GaoAC12}. Given an arbitrary positive rational number $\delta$, the $\delta$-decision problem asks if a logic formula is false or {\em $\delta$-true} (or, dually, true or {\em $\delta$-false}). The latter answer can be given, if the formula {\em would be true} under $\delta$-bounded numerical changes on its syntactic form~\cite{DBLP:conf/lics/GaoAC12}. The $\delta$-decision problem is decidable for bounded first-order sentences over the real numbers with arbitrary Type 2 computable functions. Type 2 computable functions~\cite{CAbook} are essentially real functions that can be approximated numerically. They cover almost all functions that can occur in realistic hybrid systems, such as polynomials, trigonometric functions, and solutions of Lipschitz-continuous ODEs. The goal of this paper is to develop a new framework for solving bounded reachability problems for hybrid systems based on solving $\delta$-decisions. We prove that this framework makes bounded reachability of hybrid systems a much more mathematically tractable problem and show that our practical implementation can handle highly nonlinear hybrid systems.

The framework of {\em $\delta$-complete analysis} consists of techniques that perform verification and allow bounded errors on the safe side. For bounded reachability problems, $\delta$-complete analysis aims to find one of the following answers:
\begin{itemize}
\item {\sf safe} (bounded): The system does not violate the safety property within a bounded period of time and a bounded number of discrete mode changes.
\item {\sf {$\delta$-unsafe}}: The system would violate the safety property under some $\delta$-bounded numerical perturbations.
\end{itemize}
Thus, when the answer is {\sf safe}, no error is involved. On the other hand, a system that is {\sf $\delta$-unsafe} would violate the safety property under bounded numerical perturbations. Realistic hybrid systems interact with the physical world and it is impossible to avoid slight perturbations. Thus, {\sf $\delta$-unsafe} systems should indeed be regarded as unsafe, under reasonable choices of $\delta$. Note that such robustness problems can not be discovered by solving the precise decision problem, and the use of $\delta$-decisions strengthens the verification results.

$\delta$-Complete reachability analysis reduces verification problems to $\delta$-decision problems of formulas over the reals. It follows from $\delta$-decidability of these formulas~\cite{DBLP:conf/lics/GaoAC12} that $\delta$-complete reachability analysis of a wide range of nonlinear hybrid systems is decidable. Such results stand in sharp contrast to the standard high undecidability of bounded reachability for simple hybrid systems.

We emphasize that the new framework is immediately practical. We implemented the techniques in our open-source tool {\sf dReach} based on our nonlinear SMT solver {\sf dReal}~\cite{DBLP:conf/cade/GaoKC13}. In our previous work, we have shown the underlying solver scales on nonlinear systems~\cite{DBLP:conf/fmcad/GaoKC13}. The tool successfully verified safety properties of various nonlinear models that are beyond the scope of existing tools.

The paper is organized as follows. After a short review of $\delta$-decidability, we show how to represent hybrid systems with $\lrf$-formulas and how to interpret trajectories through semantics of the formulas in Section 2. Then we focus on bounded reachability and show the encoding in $\lrf$ in Section 3. In Section 4, we show experimental results of our open-source implementation on highly nonlinear hybrid systems, and discuss the comparison with reachable set computation techniques in Section 5 and conclude in Section 6.

\paragraph{Related Work. } Our framework can be seen as a converging point for several lines of existing work. First of all, the use of logic formulas to express model checking of hybrid systems dates back to~\cite{DBLP:conf/hybrid/AlurCHH92,DBLP:journals/entcs/AudemardBCS05}, where formulas with linear arithmetic over the reals are used.  The lack of an appropriate logic for encoding nonlinear systems beyond real arithmetic has been a major bottleneck in this direction. Second, the realization that robustness assumptions help reduce verification complexity as been realized frequently. Franzle's work~\cite{DBLP:conf/csl/Franzle99} was among the first to recognize that verification problems are more tractable when robustness is assumed for polynomial hybrid systems. The direction was continued with more positive results such as~\cite{Ratschan10}. These works present theoretical results that do not directly translate to practical solving techniques, and the results are sensitive to the definitions. For instance, it is also shown in \cite{DBLP:conf/hybrid/HenzingerR00} that a slightly different notion of robustness and noise does not improve the theoretical properties. We focus on formulating a framework that directly corresponds to practical solving techniques, and the positive theoretical results follow naturally at the same time. There has also been much recent work on using constraint solving techniques for solving hybrid systems~\cite{DBLP:journals/jlp/FranzleTE10,DBLP:conf/icons/HerdeEFT08,DBLP:conf/cav/GulwaniT08,DBLP:conf/rtss/ChenAS12}, as well as solving frameworks that exploit robustness properties of the systems~\cite{DBLP:conf/rtss/PrabhakarVVD09,DBLP:conf/hybrid/HuangM12}. These methods can all handle nonlinear dynamics to certain degrees (mostly polynomial systems, with the exception of~\cite{DBLP:conf/rtss/ChenAS12} which we will mention again in the experiments). We aim to extend these works to a most broad class of nonlinear hybrid systems, and provide precise correctness guarantees. We also provide an open-source implementation that scales well on highly nonlinear systems that arise in practical applications.

%The paper is organized as follows. In Section~\ref{review} we review the first-order language $\lrf$ and $\delta$-decidability results, and define $\lrf$-representations of hybrid systems in Section~\ref{language}. We then define and study the bounded $\delta$-reachability problem in Section~\ref{main}, and show experimental results in Section~\ref{experiments}.

%\paragraph{Related Work} Franzle's work.

\section{$\lrf$-Representations of Hybrid Automata}\label{review}

\subsection{$\lrf$-Formulas and $\delta$-Decidability}
We will use a logical language over the real numbers that allows arbitrary {\em computable real functions}~\cite{CAbook}. We write $\lrf$ to represent this language. Intuitively, a real function is computable if it can be numerically simulated up to an arbitrary precision. For the purpose of this paper, it suffices to know that almost all the functions that are needed in describing hybrid systems are Type 2 computable, such as polynomials, exponentiation, logarithm, trigonometric functions, and solution functions of Lipschitz-continuous ordinary differential equations.

More formally, $\lrf = \langle \mathcal{F}, > \rangle$ represents the first-order signature over the reals with the set $\mathcal{F}$ of computable real functions, which contains all the functions mentioned above. Note that constants are included as 0-ary functions. $\lrf$-formulas are evaluated in the standard way over the structure $\mathbb{R}_{\mathcal{F}}= \langle \mathbb{R}, \mathcal{F}^{\mathbb{R}}, >^{\mathbb{R}}\rangle$. It is not hard to see that  we can put any $\lrf$-formula in a normal form, such that its atomic formulas are of the form $t(x_1,...,x_n)>0$ or $t(x_1,...,x_n)\geq 0$, with $t(x_1,...,x_n)$ composed of functions in $\mathcal{F}$. To avoid extra preprocessing of formulas, we can explicitly define $\mathcal{L}_{\mathcal{F}}$-formulas as follows.
\begin{definition}[$\lrf$-Formulas]
Let $\mathcal{F}$ be a collection of computable real functions. We define:
\begin{align*}
t& := x \; | \; f(t(\vec x)), \mbox{ where }f\in \mathcal{F} \mbox{ (constants are 0-ary functions)};\\
\varphi& := t(\vec x)> 0 \; | \; t(\vec x)\geq 0 \; | \; \varphi\wedge\varphi
\; | \; \varphi\vee\varphi \; | \; \exists x_i\varphi \; |\; \forall x_i\varphi.
\end{align*}
In this setting $\neg\varphi$ is regarded as an inductively defined operation
which replaces atomic formulas $t>0$ with $-t\geq 0$, atomic formulas $t\geq 0$
with $-t>0$, switches $\wedge$ and $\vee$, and switches $\forall$ and $\exists$.
\end{definition}
\begin{definition}[Bounded $\lrf$-Sentences]
We define the bounded quantifiers $\exists^{[u,v]}$ and $\forall^{[u,v]}$ as
$\exists^{[u,v]}x.\varphi =_{df}\exists x. ( u \leq x \land x \leq v \wedge
\varphi)$ and $
\forall^{[u,v]}x.\varphi =_{df} \forall x. ( (u \leq x \land x \leq v)
\rightarrow \varphi)$
where $u$ and $v$ denote $\lrf$ terms, whose variables only
contain free variables in $\varphi$ excluding $x$. A {\em bounded $\lrf$-sentence} is
$$Q_1^{[u_1,v_1]}x_1\cdots Q_n^{[u_n,v_n]}x_n\;\psi(x_1,...,x_n),$$
where $Q_i^{[u_i,v_i]}$ are bounded quantifiers, and $\psi(x_1,...,x_n)$ is
quantifier-free.
\end{definition}
\begin{definition}[$\delta$-Variants]\label{variants}
Let $\delta\in \mathbb{Q}^+\cup\{0\}$, and $\varphi$ an
$\lrf$-formula
$$\varphi: \ Q_1^{I_1}x_1\cdots Q_n^{I_n}x_n\;\psi[t_i(\vec x, \vec y)>0;
t_j(\vec x, \vec
y)\geq 0],$$ where $i\in\{1,...k\}$ and $j\in\{k+1,...,m\}$. The {\em
$\delta$-weakening} $\varphi^{\delta}$ of $\varphi$ is
defined as the result of replacing each atom $t_i > 0$ by $t_i >
-\delta$ and $t_j \geq 0$ by $t_j \geq -\delta$:
$$\varphi^{\delta}:\ Q_1^{I_1}x_1\cdots Q_n^{I_n}x_n\;\psi[t_i(\vec x, \vec
y)>-\delta; t_j(\vec x,
\vec y)\geq -\delta].$$
It is clear that $\varphi\rightarrow\varphi^{\delta}$~(see \cite{DBLP:conf/lics/GaoAC12}).
\end{definition}
In~\cite{DBLP:conf/lics/GaoAC12,DBLP:conf/cade/GaoAC12}, we have proved that the following $\delta$-decision problem is decidable, which is the basis of our framework.
\begin{theorem}[$\delta$-Decidability]\label{delta-decide} Let $\delta\in\mathbb{Q}^+$ be
arbitrary. There is an algorithm which, given any bounded $\lrf$-sentence $\varphi$,
correctly returns one of the following two answers:
\begin{itemize}
\item $\delta$-$\mathsf{True}$: $\varphi^{\delta}$ is true.
\item $\mathsf{False}$: $\varphi$ is false.
\end{itemize}
When the two cases overlap, either answer is correct.
\end{theorem}
\begin{theorem}[Complexity]\label{compmain}
Let $S$ be a class of $\lrf$-sentences, such that for any $\varphi$ in $S$, the terms in $\varphi$ are in Type 2 complexity class $\mathsf{C}$. Then, for any $\delta\in \mathbb{Q}^+$, the $\delta$-decision problem for bounded $\Sigma_n$-sentences in $S$ is in $\mathsf{(\Sigma_n^P)^C}$.
\end{theorem}
\subsection{{\large$\lrf$}-Representations and Hybrid Trajectories}\label{language}
We first show that $\lrf$-formulas can concisely represent hybrid automata.
\begin{definition}[$\lrf$-Representations of Hybrid Automata]\label{lrf-definition}
A hybrid automaton in $\lrf$-representation is a tuple
\begin{multline*}
H = \langle X, Q, \{{\flow}_q(\vec x, \vec y, t): q\in Q\},\{\inv_q(\vec x): q\in Q\},\\
\{\jump_{q\rightarrow q'}(\vec x, \vec y): q,q'\in Q\},\{\init_q(\vec x): q\in Q\}\rangle
\end{multline*}
where $X\subseteq \mathbb{R}^n$ for some $n\in \mathbb{N}$, $Q=\{q_1,...,q_m\}$ is a finite set of modes, and the other components are finite sets of quantifier-free $\lrf$-formulas.
\end{definition}
\begin{notation}
For any hybrid system $H$, we write $X(H)$, $\flow(H)$, etc. to denote its corresponding components.
\end{notation}
Almost all hybrid systems studied in the existing literature can be defined by restricting the set of functions $\mathcal{F}$ in the signature. For instance,

\begin{example}[Linear and Polynomial Hybrid Automata] Let $\mathcal{F}^{\mathrm{lin}} = \{+\}\cup \mathbb{Q}$ and $\mathcal{F}^{\mathrm{poly}}=\{\times\}\cup\mathcal{F}^{\mathrm{lin}}$. Rational numbers are considered as 0-ary functions. In existing literature, $H$ is a {\em linear hybrid automaton} if it has an $\mathcal{L}_{\mathbb{R}_{\mathcal{F}^{\mathrm{lin}}}}$-representation, and a {\em polynomial hybrid automaton} if it has an $\mathcal{L}_{\mathbb{R}_{\mathcal{F}^{\mathrm{poly}}}}$-representation.
\end{example}

\begin{example}[Nonlinear Bouncing Ball]
The bouncing ball is a standard hybrid system model. Its nonlinear version (with air drag) can be $\lrf$-represented in the following way:
\begin{itemize}
\item $X = \mathbb{R}^2$ and $Q = \{q_u, q_d\}$. We use $q_u$ to represent bounce-back mode and $q_d$ the falling mode.
\item $\flow = \{\flow_{q_u}(x_0, v_0, x_t, v_t, t), \flow_{q_d}(x_0, v_0, x_t, v_t, t)\}$. We use $x$ to denote the height of the ball and $v$ its velocity. Instead of using time derivatives, we can directly write the flows as integrals over time, using $\lrf$-formulas:
\begin{itemize}
\item $\flow_{q_u}(x_0, v_0, x_t, v_t, t)$ defines the dynamics in the bounce-back phase:
$$(x_t = x_0 + \int_0^{t} v(s) ds) \wedge (v_t = v_0 + \int_0^t g(1-\beta v(s)^2) ds)$$
\item $\flow_{q_d}(x_0, v_0, x_t, v_t, t)$ defines the dynamics in the falling phase:
$$(x_t = x_0 + \int_0^{t} v(s) ds) \wedge (v_t = v_0 + \int_0^t g(1+\beta v(s)^2) ds)$$
\end{itemize}where
$\beta$ is a constant. Again, note that the integration terms define Type 2 computable functions.
\item $\jump = \{\jump_{q_u \rightarrow q_d} (x, v, x', v'), \jump_{q_d \rightarrow q_u} (x, v, x', v')\}$ where
\begin{itemize}
 \item $\jump_{q_u \rightarrow q_d} (x, v, x', v')$ is $(v= 0 \wedge x' = x \wedge v' = v)$.
\item $\jump_{q_d \rightarrow q_u} (x, v, x', v')$ is $(x= 0 \wedge v' = \alpha v\wedge x'=x)$,  for some constant $\alpha$.
\end{itemize}
\vspace{0.1cm}
\item $\init_{q_d}$ is $(x=10 \wedge v=0)$ and $\init_{q_u}$ is $\bot$.
\item $\inv_{q_d}$ is $(x>=0 \wedge v>=0)$ and $\inv_{q_u}$ is $(x>=0 \wedge v<=0)$.
\end{itemize}
\end{example}

Trajectories of hybrid systems combine continuous flows and discrete jumps. This motivates the use of a hybrid time domain, with which we can keep track of both the discrete changes and the duration of each continuous flow. A hybrid time domain is a sequence of closed intervals on the real line, and a hybrid trajectory is a mapping from the time domain to the Euclidean space. Formally, we use the following definition given by Davoren in \cite{DBLP:conf/hybrid/Davoren09}:
\begin{definition}[Hybrid Time Domains and Hybrid Trajectories~\cite{DBLP:conf/hybrid/Davoren09}]
A {\em hybrid time domain} is a subset of
 $\mathbb{N}\times \mathbb{R}$ of the form
$$T_m=\{(i, t): i<m \mbox{ and } t\in [t_i, t_i']\mbox{ or }[t_i, +\infty)\},$$
where $m\in \mathbb{N}\cup\{+\infty\}$, $\{t_i\}_{i=0}^m$ is an
increasing sequence in $\mathbb{R}^+$, $t_0= 0$, and $t_i'=t_{i+1}$. When $X\subseteq\mathbb{R}^n$ is an Euclidean space and $T_m$ a hybrid
time domain, a {\em hybrid trajectory} is a {\em continuous} mapping $\xi: T_m\rightarrow X.$ We can write the time domain $T_m$ of $\xi$ as $T(\xi)$.
 \end{definition}
We can now define trajectories of hybrid automata. To link hybrid trajectories with automata, we need a labeling function $\sigma_{\xi,H}(i)$ that maps each step $i$ in the hybrid trajectory to an appropriate discrete mode in $H$, and make sure that the $\flow, \jump, \inv, \init$ conditions are satisfied.
\begin{definition}[Trajectories of Hybrid Automata]\label{trajec}
Let $H$ be a hybrid automaton, $T_m$ a hybrid domain, and $\xi: T_m\rightarrow X$ a hybrid trajectory.
We say that $\xi$ is {\em a trajectory of $H$ of discrete depth $m$}, written as $\xi\in \llbracket H \rrbracket$, if there
exists a {\em labeling function} $\sigma_{\xi,H}: \mathbb{N}\rightarrow Q$ such
that:
\begin{itemize}
\item For some $q\in Q$, $\sigma_{\xi,H}(0) = q$ and $\mathbb{R}_{\mathcal{F}}\models \init_{q}(\xi(0,0))$.
\item For any $(i, t)\in T_m$, $\mathbb{R}_{\mathcal{F}}\models \inv_{\sigma_{\xi,H}(i)} (\xi(i,t))$.
\item For any $(i,t)\in T_m$,
\begin{itemize}
\item When $i=0$, $\mathbb{R}_{\mathcal{F}}\models\flow_{q_0}(\xi(0,0), \xi(0,t), t).$
\item When $i = k+1$, where $0<k+1<m$,
\begin{eqnarray*}
& &\mathbb{R}_{\mathcal{F}}\models\flow_{\sigma^H_{\xi}(k+1)}(\xi(k+1, t_{k+1}), \xi(k+1, t), (t - t_{k+1})), and\\
& &\mathbb{R}_{\mathcal{F}}\models \jump_{\sigma_{\xi,H}(k)\rightarrow\sigma_{\xi,H}(k+1)}(\xi(k, t_k'), \xi(k+1,t_{k+1})).
\end{eqnarray*}
\end{itemize}
\end{itemize}
\end{definition}
The definition is straightforward. In each mode, the system flows continuously following the dynamics defined by $\flow_q$. Note that $(t-t_k)$ is the actual duration in the $k$-th mode. When a switch between two modes is performed, it is required that $\xi(k+1, t_{k+1})$ is updated from the exit value $\xi(k, t_k')$ in the previous mode, following the jump conditions.
%We write $\llbracket H\rrbracket$ to denote the set of all the trajectories of $H$.
\begin{remark}[$\jump$ vs $\inv$] The jump conditions specify when $H$ {\em may}  switch to another mode. The invariants (when violated) specify when $H$ {\em must} switch to another mode. They will require different logical encodings.
\end{remark}
Note that we gave no restriction on the formulas that can be used for describing hybrid automata in Definition~\ref{lrf-definition}. A minimal requirement is that the $\flow$ predicates should define continuous trajectories over time, namely:
\begin{definition}[Well-Defined Flow Predicates]
Let $\flow(\vec x, \vec y, t)$ be a flow predicate for a hybrid automaton $H$. We say the flow predicate is {\em well-defined}, if for all tuples $(\vec a,\vec b, \tau)\in X(H)\times X(H)\times \mathbb{R}^{\geq 0}$ such that $\mathbb{R}\models\flow(\vec a,\vec b, \tau)$, there exists a continuous function $\eta:[0,\tau]\rightarrow X$ such that $\eta(0) = \vec a$, $\eta(\tau) = \vec b$, and for all $t'\in [0,\tau]$, we have $\mathbb{R}\models\flow(\vec a, \eta(t), t)$. We say $H$ is {\em well-defined} if all its flow predicates are well-defined.
\end{definition}
This definition requires that we can always construct a trajectory from the end points and the initial points that satisfy a flow predicate. Flows that are defined using differential equations, differential inclusions, and explicit continuous mappings all satisfy this condition. Thus, from now on our discussion of hybrid automata assume their well-definedness. 
\subsection{$\delta$-Perturbations}
We can now define $\delta$-perturbations on hybrid automata directly through perturbations on the logic formulas in their $\lrf$-representations. For any set $S$ of $\lrf$-formulas, we write $S^{\delta}$ to denote the set containing the $\delta$-perturbations of all elements of $S$.
\begin{definition}[$\delta$-Weakening of Hybrid Automata] Let $\delta\in\mathbb{Q}^+\cup\{0\}$ be arbitrary. Suppose
$$H = \langle X, Q, \flow, \jump, \inv, \init\rangle$$
is an $\lrf$-representation of hybrid system $H$. The {\em $\delta$-weakening} of $H$ is
$$H^{\delta} = \langle X, Q, \flow^{\delta}, \jump^{\delta}, \inv^{\delta}, \init^{\delta}\rangle$$
which is obtained by weakening all formulas in the $\lrf$-representations of $H$.
\end{definition}

\begin{example}
The $\delta$-weakening of the bouncing ball automaton is obtained by weakening the formulas in its description. For instance, $\flow_{q_u}^{\delta}(x_0, v_0, x_t, v_t, t)$ is
{$$|x_t - (x_0 + \int_0^{t} v(s) ds)|\leq \delta \wedge |v_t - (v_0 + \int_0^t g(1-\beta v(s)^2) ds))|\leq \delta$$}
and $\jump_{q_d \rightarrow q_u}^{\delta} (x, v, x', v')$ is {$$|x|\leq \delta \wedge |v' - \alpha v|\leq \delta \wedge |x'-x|\leq \delta.$$}
\end{example}
\begin{remark} It is important to note that the notion of $\delta$-perturbations is a purely syntactic one (defined on the description of hybrid systems) instead of a semantic one (defined on the trajectories). The syntactic perturbations correspond to semantic over-approximation of $H$ in the trajectory space.
\end{remark}
\begin{proposition} For any $H$ and $\delta\in\mathbb{Q}^+\cup\{0\}$, $\llbracket H\rrbracket\subseteq \llbracket H^{\delta}\rrbracket$.
\end{proposition}
\begin{proof}
Let $\xi\in \llbracket H\rrbracket$ be any trajectory of $H$. Following
Definition~\ref{variants}, for any $\lrf$ sentence $\varphi$, we have
$\varphi\rightarrow\varphi^{\delta}$. Since
$\xi$ satisfies the conditions in Definition~\ref{trajec}, after replacing each
formula by their $\delta$-weakening, we have $\xi\in \llbracket H^{\delta}\rrbracket$.
\end{proof}

\subsection{Reachability}
We can now formally state the reachability problem for hybrid automata using $\lrf$-representations and their interpretations.
\begin{definition}[Reachability]\label{reachability}\index{Reachability}
Let $H$ be an $n$-dimensional hybrid automaton, and $U$ a subset of its state
space $Q\times X$.  We say {\em $U$ is reachable by $H$}, if there exists
$\xi\in\llbracket
H \rrbracket$,such that there exists $(i,t)\in T(\xi)$ satisfying
$(\sigma^H_{\xi}(i), \xi(i,t))\in U.$
\end{definition}
The bounded reachability problem for hybrid systems is defined by restricting
the continuous time duration to a bounded interval, and the number
of discrete transitions to a finite number.
\begin{definition}[Bounded Reachability]
Let $H$ be an $n$-dimensional hybrid automaton, whose continuous state space
$X$ is a bounded subset of $\mathbb{R}^n$. Let $U$ be a subset of its state
space. Set $k\in \mathbb{N}$ and $M \in \mathbb{R}^{\geq 0}$. The {$(k,M)$-bounded
reachability problem} asks whether there exists
$\xi\in\llbracket H \rrbracket$ such that there exists $(i,t)\in T(\xi)$ with $i\leq k$, $t=
\sum_{i=0}^k t_i$ where $t_i \leq M$, and $(\sigma_{\xi}(i), \xi(i,t))\in U.$
\end{definition}
%\begin{notation}When $M\in\mathbb{R^+}$, we write $T^M$ to denote a time domain whose continuous fragments are bounded by $M$, i.e., for each $(i,t_i)\in T^M$, $t_i\leq M$.\end{notation}
\begin{remark}
By ``step", we mean the number of discrete jumps. We say $H$ can reach $U$ in $k$ steps, if there exists $\xi\in\llbracket H\rrbracket$ that contains $k$ discrete jumps, which consists of $k+1$ pieces of continuous flows in the corresponding discrete modes.
\end{remark}

In the seminal work of \cite{DBLP:conf/rex/AlurD91,DBLP:conf/hybrid/AlurCHH92}, it is already shown that the bounded reachability problem for simple classes of hybrid automata is undecidable. The goal of $\delta$-complete analysis is to bypass much of this difficulty.

\section{$\delta$-Complete Analysis for Bounded Reachability}\label{main}
\subsection{Encoding Bounded Reachability in $\lrf$}
We now define the $\lrf$-encoding of bounded reachability. The encodings are standard bounded model checking, and have been studied in existing work but without the generality of a full $\lrf$-language. As a result, some issues have not been discovered. For example, the full encoding of non-deterministic flows with invariant conditions require second-order quantification, and the first-order encoding requires additional assumptions. We will give the full $\lrf$-encodings and discuss such details.
\begin{notation}
Let $H$ be a hybrid automaton. We use $\unsafe = \{\unsafe_q:q\in Q\}$ as the $\lrf$-representation of an unsafe region in the state space of $H$. We can write $\llbracket \unsafe\rrbracket = \bigcup_{q\in Q} \llbracket \unsafe_q \rrbracket\times \{q\}$.
\end{notation}
First, we need to define a set of auxiliary formulas that will be important for ensuring that a particular mode is picked at a certain step.
\begin{definition}
Let $Q = \{q_1,...,q_m\}$ be a set of modes. For any $q\in Q$, and $i\in\mathbb{N}$, use  $b_{q}^i$ to represent a Boolean variable. We now define
$$\enforce_Q(q,i) = b^i_{q} \wedge \bigwedge_{p\in Q\setminus\{q\}}\neg b^{i}_{p}$$
$$\enforce_Q(q, q',i) = b^{i}_{q}\wedge \neg b^{i+1}_{q'} \wedge \bigwedge_{p\in Q\setminus\{q\}} \neg b^i_{p} \wedge \bigwedge_{p'\in Q\setminus\{q'\}} \neg b^{i+1}_{p'}$$
We omit the subscript $Q$ when the context is clear.\end{definition}
The use of the auxiliary of formulas will be explained when we define the full encodings of bounded reachability.
\paragraph{Systems with no invariants.} We start with the simplest case for hybrid automata with no invariants. Naturally, we say a hybrid automaton $H$ is {\em invariant-free} if $\inv_q(H) = \top$ for every $q\in Q(H)$. We define the following formula that checks whether an unsafe region is reachable after exactly $k$ steps of discrete transition in a hybrid system.
\begin{definition}[$k$-Step Reachability, Invariant-Free Case]
Suppose $H$ is invariant-free, and $U$ a subset of its state space represented by $\unsafe$. The $\lrf$-formula $\reach_{H,U}(k,M)$ is defined as:
\begin{eqnarray*}
%\reach^{k,M}(H,U) &:=&
& &\exists^X \vec x_{0} \exists^X\vec x_{0}^t\cdots \exists^X \vec x_{k}\exists^X\vec x_{k}^t\exists^{[0,M]}t_0\cdots \exists^{[0,M]}t_k.\\
& &\bigvee_{q\in Q} \Big(\init_{q}(\vec x_{0})\wedge \flow_{q}(\vec x_{0}, \vec x_{0}^t, t_0)\wedge \enforce(q,0)\Big)\\%\wedge (b_{q_i}\wedge \bigwedge_{q\neq q_i} \neg b_{q})
\wedge & & \bigwedge_{i=0}^{k-1}\bigg( \bigvee_{q, q'\in Q} \Big(\jump_{q\rightarrow q'}(\vec x_{i}^t, \vec x_{i+1})\wedge \enforce(q,q',i)\\
& & \hspace{4.7cm}\wedge\flow_{q'}(\vec x_{i+1}, \vec x_{i+1}^t, t_{i+1})\wedge \enforce(q',i+1)\Big)\bigg)\\
\wedge & & \bigvee_{q\in Q} \unsafe_q(\vec x_{k}^t).
\end{eqnarray*}
\end{definition}
Intuitively, the trajectories start with some initial state satisfying $\init_q(\vec x_{0})$ for some $q$. In each step, it follows $\flow_q(\vec x_{i}, \vec x_{i}^t, t)$ and makes a continuous flow from $\vec x_i$ to $\vec x_i^t$ after time $t$. When $H$ makes a $\jump$ from mode $q'$ to $q$, it resets variables following $\jump_{q'\rightarrow q}(\vec x_{k}^t, \vec x_{k+1})$. The auxiliary $\enforce$ formulas ensure that picking $\jump_{q\rightarrow q'}$ in the $i$-the step enforces picking $\flow_q'$ in the $(i+1)$-th step.

\paragraph{Systems with invariants and deterministic flows.} When the invariants are not trivial, we need to ensure that for all the time points along a continuous flow, the invariant condition holds. Thus, we need to universally quantify over time. This is a fact that has been previously discussed, for instance, in~\cite{DBLP:conf/fmcad/CimattiMT12}. However, if we allow nondeterministic flows, the situation is more complicated, which has not been discovered in existing work. We give the encoding for systems with only deterministic flows first, as follows:
\begin{definition}[$k$-Step Reachability, Nontrivial Invariant and Deterministic Flow]\label{br2}
Suppose $H$ contains invariants and only deterministic flow
, and $U$ a subset of its state space represented by $\unsafe$. In this case, the $\lrf$-formula $\reach_{H,U}(k,M)$ is defined as:
\begin{eqnarray*}
& &\exists^X \vec x_{0} \exists^X\vec x_{0}^t\cdots \exists^X \vec x_{k}\exists^X\vec x_{k}^t \exists^{[0,M]}t_0\cdots \exists^{[0,M]}t_k.\\
& &\bigvee_{q\in Q} \Big(\init_{q}(\vec x_{0})\wedge \flow_{q}(\vec x_{0}, \vec x_{0}^t, t_0)\wedge \enforce(q,0)\\
& &\hspace{5cm} \wedge \forall^{[0,t_0]}t\forall^X\vec x\;(\flow_{q}(\vec x_{0}, \vec x, t)\rightarrow \inv_{q}(\vec x))\Big) \\
\wedge & &\bigwedge_{i=0}^{k-1}\bigg( \bigvee_{q, q'\in Q} \Big(\jump_{q\rightarrow q'}(\vec
x_{i}^t, \vec x_{i+1})\wedge \flow_{q'}(\vec x_{i+1}, \vec x_{i+1}^t, t_{i+1})\wedge \enforce(q,q',i)\\
& & \hspace{1.5cm}\wedge\enforce(q',i+1)\wedge \forall^{[0,t_{i+1}]}t\forall^X\vec x\;(\flow_{q'}(\vec x_{i+1}, \vec x,
t)\rightarrow \inv_{q'}(\vec x)) )\Big)\bigg)\\
\wedge & &\bigvee_{q\in Q} (\unsafe_q(\vec x_{k}^t)\wedge \enforce(q,k)).
\end{eqnarray*}
\end{definition}
The extra universal quantifier for each continuous flow expresses the requirement that for all the time points between the initial and ending time point ($t\in[0,t_i+1]$) in a flow, the continuous variables $\vec x$ must take values that satisfy the invariant conditions $\inv_q(\vec x)$.

\paragraph{Systems with invariants and nondeterministic flows.} In the most general case, a hybrid system can contain non-deterministic flow: i.e., for some $q\in Q$, there exists $\vec a_0, \vec a_t, \vec a_t'\in \mathbb{R}^n$ and $t\in\mathbb{R}$ such that $\vec a_t\neq \vec a_t'$ and $\mathbb{R}\models \flow_q(\vec a_0, \vec a_t, t)$ and $\mathbb{R}\models \flow_q(\vec a_0, \vec a_t', t)$. Consequently, there is multiple possible values for the continuous variable for each time point. Different values correspond to different trajectories, and we only look for one of the trajectories that satisfies the invariant on all time points. Thus, we need to quantify over a trajectory and write $\exists \xi \forall t.\; \inv(\xi(t))$. We conjecture that, in general, this second-order quantification can not be fully reduced to a first-order expression.

In practice, the discussion of the invariant conditions in the existing work has (implicitly) assumed that the invariant condition should hold for all possible trajectories in the case of non-deterministic flow. We can formulate this assumption in the following way:
\begin{definition}[Strictly-Imposed Invariants]
We say a hybrid automaton $H$ has strictly-imposed mode invariants, if the following condition holds. Let $\flow_q(\vec x, \vec y, t)$ and $\inv_q(\vec x)$ be the flow and invariant conditions in any mode $q$ of $H$. Let $\vec a$ be an arbitrary starting point in the mode, satisfying $\inv(\vec a)$. Then, for any  $\vec b, \vec b'\in X(H)$ such that $\flow(\vec a, \vec b, \tau)$ and $\flow(\vec a, \vec b', \tau)$ are true at the same time point $\tau\in \mathbb{R}$, we have $\inv_q(\vec b)$ iff $\inv_q (\vec b')$.
\end{definition}
If this condition is true, then a witness trajectory of bounded reachability has to require that all flows satisfy the same invariants. Consequently, we can still use the encoding in Definition~\ref{br2}, which requires that all possible flows satisfy the invariants. Thus, when this condition applies, we can still use first-order encoding for reachability in the presence of non-deterministic flows.

\subsection{$\delta$-Complete Analysis of Bounded Reachability}

We now define the $\delta$-complete analysis problem and prove its decidability.
\begin{definition}
Let $H$ be a hybrid system and $U$ a subset of its state space. Suppose $U$ is represented by the $\lrf$-formula $\unsafe$. Let $k\in \mathbb{N}$ and $M\in \mathbb{R}^+$. The $\delta$-complete analysis for $(k,M)$-bounded reachability problem asks for one of the following answers:
\begin{itemize}
\item {\sf $(k,m)$-Safety:} $H$ does not reach $\llbracket\unsafe\rrbracket$ within the $(k,M)$-bound.
\item {\sf $\delta$-Unsafety:} $H^{\delta}$ reaches $\llbracket\unsafe^{\delta}\rrbracket$ within the $(k,M)$-bound.
\end{itemize}
\end{definition}
The following lemma comes from the intuitive meaning of the encodings. A proof is given in the appendix. 
\begin{lemma}\label{equiv-delta}
Let $\delta\in\mathbb{Q}^+\cup \{0\}$ be arbitrary. Suppose $H$ is a well-defined hybrid automaton with strictly-imposed invariants. Let $U$ a subset of the state space of $H$, represented by the set  $\unsafe$ of $\lrf$-formulas. Let $\reach_{H,U}(k,M)$ be the $\lrf$-formula encoding $(k,M)$-bounded reachability of $H$ with respect to $U$.  We always have that $\mathbb{R}\models(\reach_{H,U}(k,M))^{\delta}$ iff there exists a trajectory $\xi\in\llbracket H^{\delta}\rrbracket$ such that for some $(k,t)\in T(\xi)$, where $0\leq t\leq M$, $(\xi(k, t), \sigma_{\xi}(k))\in \llbracket \unsafe^{\delta}\rrbracket$.
\end{lemma}
Now we can show that $\delta$-complete analysis for bounded reachability problems is decidable for general $\lrf$-representable hybrid systems.
\begin{theorem}[Decidability]
Let $\delta\in \mathbb{Q}^+$ be arbitrary. There exists an algorithm such that, for any bounded well-defined hybrid automaton $\lrf$-represented by $H$ with strictly imposed invariants, and any unsafe region $U$ $\lrf$-represented by $\unsafe$, correctly performs $\delta$-complete analysis for $(k,M)$-bounded reachability for $H$, for any $k\in \mathbb{N}, M\in \mathbb{R}^+$.
\end{theorem}
\begin{proof}
We need to show that there is an algorithm that correctly returns one of the following:
\begin{itemize}
\item $H$ does not reach $\llbracket\unsafe\rrbracket$ within the $(k,M)$-bound.
\item $H^{\delta}$ reaches $\llbracket\unsafe^{\delta}\rrbracket$ within the $(k,M)$-bound.
\end{itemize}
To do this, we only need to solve the $\delta$-decision problem of $\reach_{H,U}(i,M)$ for $0\leq i\leq k$. We obtain either $\reach_{H,U}(i,M)$ is false for all such $i$, or is $\delta$-true for some $i$, then:
\begin{itemize}
\item Suppose $\reach_{H,U}(i,M)$ is false for all $i$. Then we know that for any $i\leq k$, $\reach_{H,U}(i,M)$ is false. Using Lemma~\ref{equiv-delta} for the special case $\delta=0$, we know that there does not exist a trajectory $\xi\in\llbracket H\rrbracket$ that can reach $U$ within $i$ steps, and consequently the system is safe within the $(k,M)$-bound.
\item Suppose $\reach_{H,U}(i,M)$ is $\delta$-true for some $i$. We know that there exists $i\leq k$ such that $\reach^{\delta}_{H,U}(i,M)$ is true. Using Lemma~\ref{equiv-delta} for $\delta\in\mathbb{Q}^+$, we know that there exists a trajectory $\xi\in\llbracket H^{\delta}\rrbracket$ that can reach the region represented by $\unsafe^{\delta}$ in $i$-steps, i.e., within the $(k,M)$-bound.\qed
\end{itemize}
\end{proof}
From the structures of the $\lrf$-formulas encoding $\delta$-reachability, we can obtain the following complexity results of the reachability problems.
\begin{theorem}[Complexity]
Suppose all the $\lrf$-terms in the description of $H$ and $U$ are in complexity class $\mathsf{C}$. Then deciding the $(k,M)$-bounded $\delta$-reachability problem is in
\begin{itemize}
\item $\np^{\mathsf{C}}$ for an invariant-free $H$;
\item $(\Sigma_2^P)^{\mathsf{C}}$ for an $H$ with strictly-imposed nontrivial invariants.
\end{itemize}
\end{theorem}

\begin{corollary}
For linear and polynomial hybrid automata, $\delta$-complete bounded reachability analysis ranges from being $\mathsf{NP}$-complete to $\mathsf{\Sigma_2^P}$-complete for the three cases. For hybrid automata that can be $\lrf$-represented with whose $\mathcal{F}$ contains the set of ODEs defined $\mathsf{P}$-computable right-hand side functions, the problem is $\mathsf{PSPACE}$-complete.
\end{corollary}
The results come from the fact that the complexity of polynomials is in $\mathsf{P}$, and the set of ODEs in questions are $\mathsf{PSPACE}$-complete.
\begin{remark}
The complexity results indicate that the worst-case running time of the analysis is exponential in all the input parameters. In particular, the worst-case running time grows exponentially with the $\delta$ and the size of the domains. We need to use efficient decision procedures to manage this complexity.
\end{remark}

\section{Experiments}

Our tool {\sf dReach} implements the techniques presented in the
paper. The tool is built on several existing packages,including {\sf
  opensmt}~\cite{DBLP:conf/tacas/BruttomessoPST10} for the general
DPLL(T) framework, {\sf
  realpaver}~\cite{DBLP:journals/toms/GranvilliersB06} for ICP, and
{\sf CAPD}~\cite{capd} for computing interval-enclosures of ODEs. The
tool is open-source at~\url{http://dreal.cs.cmu.edu/dreach.html}. All
benchmarks and data shown here are also available on the tool
website.All experiments were conducted on a machine with a 3.4GHz
octa-core Intel Core i7-2600 processor and 16GB RAM, running 64-bit
Ubuntu 12.04LTS. Table~\ref{tbl:exp} is a summary of the running time
of the tool on various hybrid system models which we explain below.

\paragraph{Atrial Fibrillation.} We studied the Atrial Fibrillation model as developed in~\cite{DBLP:conf/cav/GrosuBFGGSB11}. The model has four discrete control locations, four state variables, and nonlinear ODEs. A typical set of ODEs in the model is:
\begin{eqnarray*}
\frac{du}{dt} &=& e + (u-\theta_v)(u_u-u ) v g_{fi} + wsg_{si}-g_{so}(u)\\
\frac{ds}{dt} &=& \displaystyle\frac{g_{s2}}{(1+\exp(-2k(u-us)))} -  g_{s2}s\\
\frac{dv}{dt} &=& -g_v^+\cdot v \hspace{1cm} \frac{dw}{dt} = -g_w^+\cdot w
\end{eqnarray*}
The exponential term on the right-hand side of the ODE is the sigmoid function, which often appears in modelling biological switches.
\paragraph{Prostate Cancer Treatment.} The Prostate Cancer Treatment model~\cite{bing} exhibits more nonlinear ODEs. The reachability questions are
\begin{eqnarray*}
\frac{dx}{dt} &=& (\alpha_x
(k_1+(1-k_1)\frac{z}{z+k_2}-\beta_x( (1-k_3)\frac{z}{z+k_4}+k_3)) - m_1(1-\frac{z}{z_0}))x + c_1 x\\
\frac{dy}{dt} &=& m_1(1-\frac{z}{z_0})x+(\alpha_y (1- d\frac{z}{z_0}) - \beta_y)y+c_2y\\
\frac{dz}{dt} &=& \frac{-z}{\tau} + c_3z\\
\frac{dv}{dt} &=& (\alpha_x
(k_1+(1-k_1)\frac{z}{z+k_2}-\beta_x(k_3+(1-k_3)\frac{z}{z+k_4}))\\
& &- m_1(1-\frac{z}{z_0}))x + c_1 x + m_1(1-\frac{z}{z_0})x+(\alpha_y (1- d\frac{z}{z_0}) - \beta_y)y+c_2y
\end{eqnarray*}
\paragraph{Electronic Oscillator.} The EO model represents an electronic oscillator model that contains nonlinear ODEs such as the following:
\begin{eqnarray*}
\frac{dx}{dt} &=& - ax \cdot sin(\omega_1 \cdot \tau)\\
\frac{dy}{dt} &=& - ay \cdot sin( (\omega_1 + c_1) \cdot \tau) \cdot sin(\omega_2)\cdot 2\\
\frac{dz}{dt} &=& - az \cdot sin( (\omega_2 + c_2) \cdot \tau) \cdot cos(\omega_1)\cdot 2\\
\frac{\omega_1}{dt} &=& - c_3\cdot \omega_1\ \ \ \frac{\omega_2}{dt} = -c_4\cdot\omega_2\ \ \ \frac{d\tau}{dt} = 1
\end{eqnarray*}
\paragraph{Quadcopter Control.} We developed a model that contains the full dynamics of a quadcopter. We use the model to solve control problems by answering reachability questions. A typical set of the differential equations are the following:
\begin{eqnarray*}
\frac{\mathrm{d}\omega_x}{\mathrm{d}t} &=& L\cdot k\cdot (\omega_1^2 - \omega_3^2)(1/I_{xx})-(I_{yy} - I_{zz})\omega_y\omega_z/I_{xx}\\
\frac{\mathrm{d}\omega_y}{\mathrm{d}t} &=& L\cdot k\cdot(\omega_2^2 - \omega_4^2)(1/I_{yy})-(I_{zz} - I_{xx})\omega_x\omega_z/I_{yy}\\
\frac{\mathrm{d}\omega_z}{\mathrm{d}t} &=& b\cdot(\omega_1^2 - \omega_2^2 + \omega_3^2 - \omega_4^2)(1/I_{zz})-(I_{xx} - I_{yy})\omega_x\omega_y/I_{zz}\\
\frac{\mathrm{d}\phi}{\mathrm{d}t} &=& \omega_x + \displaystyle{\frac{\sin\left(\phi\right) \sin\left(\theta\right)}{{\left(\frac{\sin\left(\phi\right)^{2} \cos\left(\theta\right)}{\cos\left(\phi\right)} + \cos\left(\phi\right) \cos\left(\theta\right)\right)} \cos\left(\phi\right)}}\omega_y + \displaystyle\frac{\sin\left(\theta\right)}{\frac{\sin\left(\phi\right)^{2} \cos\left(\theta\right)}{\cos\left(\phi\right)} + \cos\left(\phi\right) \cos\left(\theta\right)}\omega_z\\
\frac{\mathrm{d}\theta}{\mathrm{d}t} &=& -(\displaystyle\frac{\sin\left(\phi\right)^{2} \cos\left(\theta\right)}{{\left(\frac{\sin\left(\phi\right)^{2} \cos\left(\theta\right)}{\cos\left(\phi\right)}\omega_y + \cos\left(\phi\right) \cos\left(\theta\right)\right)} \cos\left(\phi\right)^{2}} + \frac{1}{\cos\left(\phi\right)})\omega_y\\
& &\hspace{5cm}-\displaystyle\frac{\sin\left(\phi\right) \cos\left(\theta\right)}{{\left(\frac{\sin\left(\phi\right)^{2} \cos\left(\theta\right)}{\cos\left(\phi\right)} + \cos\left(\phi\right) \cos\left(\theta\right)\right)} \cos\left(\phi\right)}\omega_z \\
\frac{\mathrm{d}\psi}{\mathrm{d}t} &=& \displaystyle\frac{\sin\left(\phi\right)}{{\left(\frac{\sin\left(\phi\right)^{2} \cos\left(\theta\right)}{\cos\left(\phi\right)} + \cos\left(\phi\right) \cos\left(\theta\right)\right)} \cos\left(\phi\right)}\omega_y + \displaystyle\frac{1}{\frac{\sin\left(\phi\right)^{2} \cos\left(\theta\right)}{\cos\left(\phi\right)} + \cos\left(\phi\right) \cos\left(\theta\right)}\omega_z\\
\frac{\mathrm{d}{xp}}{\mathrm{d}t} &=& (1/m)(\sin(\theta)\sin(\psi)k(\omega_1^2 + \omega_2^2 +\omega_3^2+\omega_4^2) - k\cdot d\cdot{xp})\\
\frac{\mathrm{d}{yp}}{\mathrm{d}t} &=& (1/m)(-\cos(\psi)\sin(\theta)k(\omega_1^2 + \omega_2^2 +\omega_3^2+\omega_4^2) - k\cdot d\cdot{yp})\\
\frac{\mathrm{d}{zp}}{\mathrm{d}t} &=& (1/m)(-g-\cos(\theta)k(\omega_1^2 + \omega_2^2 +\omega_3^2+\omega_4^2) - k\cdot d\cdot{zp}\\
\frac{\mathrm{d}x}{\mathrm{d}t} &=& {xp}, \frac{\mathrm{d}y}{\mathrm{d}t} = {yp}, \frac{\mathrm{d}z}{\mathrm{d}t} = {zp}
\end{eqnarray*}

\newcommand{\hmodel}[2]{\href{http://dreal.cs.cmu.edu/#1}{#2}}
{\small
\begin{table}[!th]
  \centering
  \small
  \begin{tabular}{l|r|r|r|r|r|r|r|r}
    \hline
    \hline
    Benchmark    & \#Mode& \#Depth & \#ODEs & \#Vars  & Delta  & Result       & Time(s) & Trace \\
    \hline
    \hline
      AF-GOOD & 4     & 3        & 20     & 53      & 0.001     & SAT &  0.425    & 793K     \\
       AF-BAD & 4     & 3        & 20     & 53      & 0.001     & UNSAT &  0.074    & ---      \\
  AF-TO1-GOOD & 4     & 3        & 24     & 62      & 0.001     & SAT &  2.750    & 224K     \\
   AF-TO1-BAD & 4     & 3        & 24     & 62      & 0.001     & UNSAT &  5.189    & ---     \\
  AF-TO2-GOOD & 4     & 3        & 24     & 62      & 0.005     & SAT &  3.876    & 553K     \\
   AF-TO2-BAD & 4     & 3        & 24     & 62      & 0.001     & UNSAT &  8.857    & ---     \\
 AF-TSO1-TSO2 & 4     & 3        & 24     & 62      & 0.001     & UNSAT &  0.027    & ---     \\
       AF8-K7 & 8     & 7        & 40     & 101     & 0.001     & SAT & 10.478   & 3.8M      \\
      AF8-K23 & 8     & 23       & 40     & 293     & 0.001     & SAT & 135.29   & 11M      \\
    \hline
    \hline
    EO-K2  & 3     & 2        & 18     & 48      & 0.01    & SAT & 3.144    & 1.9M      \\
    EO-K11 & 3     & 11       & 99     & 174     & 0.01    & UNSAT & 0.969    & ---       \\
    \hline
    \hline
    QUAD-K1  & 2   & 1          & 34     & 89      & 0.01      & SAT & 2.386 &  10M \\
    QUAD-K2  & 2   & 2          & 34     & 125     & 0.01      & SAT & 4.971 &  13M \\
    QUAD-K3  & 4   & 3          & 68     & 161     & 0.01      & SAT & 13.755 & 42M \\
    QUAD-K3U & 4   & 3          & 68     & 161     & 0.01      & UNSAT & 2.846 & --- \\
    \hline
    \hline
    CT       & 2   & 2         & 10      & 41      & 0.005     & SAT & 345.84 & 3.1M\\
    CT       & 2   & 2         & 10      & 41      & 0.002     & SAT & 362.84 & 3.1M\\
    \hline
    \hline
    BB-K10 & 2     & 10       & 22     & 66      & 0.01        & SAT & 8.057     & 123K  \\
    BB-K20 & 2     & 20       & 42     & 126     & 0.01        & SAT & 39.196    & 171K  \\
    \hline
    \hline
  \end{tabular}
  \caption{\small
    \#Mode = Number of modes in the hybrid system,
    \#Depth = Unrolling depth,
    \#ODEs = Number of ODEs in the unrolled formula,
    \#Vars = Number of variables in the unrolled formula,
    Result = Bounded Model Checking Result (delta-SAT/UNSAT)
    Time = CPU time (s),
    Trace = Size of the ODE trajectory,
    AF = Atrial Filbrillation,
    EO = Electronic Oscillator,
    QUAD = Quadcopter Control,
    CT = Cancer Treatment,
    BB = Bouncing Ball with Drag.
    % TIMES = Solving time in seconds, TO = Timeout (30min), PC = Proof
    % Checked, #PA = Number of proved axioms, #SP = Number of subproblems
    % generated by proof checking, TIMEPC = Proof-checking time in seconds, #D =
    % Number of iteration depth required in proof checking
}\label{tbl:exp}
\end{table}
}
%%% Local Variables:
%%% mode: latex
%%% TeX-master: "delta_reachability.tex"
%%% End:

\paragraph{Room for Improvements.} We aim to provide an open-source
framework that allows much more optimization. In particular, while we
can solve highly nonlinear models that are beyond the scope of other
existing tools, there are simpler examples that other tools perform
better. For instance, the Flow* tool~\cite{DBLP:conf/rtss/ChenAS12}
can efficiently compute a tight enclosure of the following system,
while our tool does not terminate in reasonable time:
\begin{eqnarray*}
dx/dt &=& -9(x - 2) - 7(y + 2) + (z - 1) + 0.2(x - 2)(y + 2) \\
& &\hspace{2cm}+0.1(y + 2)(z - 1) + 0.1(x - 2)(z - 1) + 0.5(z - 1)^2\\
dy/dt &=& 6(x - 2) + 4(y + 2) + z - 1\\
dz/dt &=& 3(x - 2) + 2(y + 2) - 2.5(z - 1)
\end{eqnarray*}
The reason is that the CAPD package that we use for verified integration of ODE blows up on this set of equations. However, our framework can integrate any reachable set computation tool, in replace of CAPD, for computing pruning on continuous flows. We remark on this in the next section.
\begin{figure}[h!]
\begin{center}
\includegraphics[width=1.0\textwidth, height=0.6\textwidth]{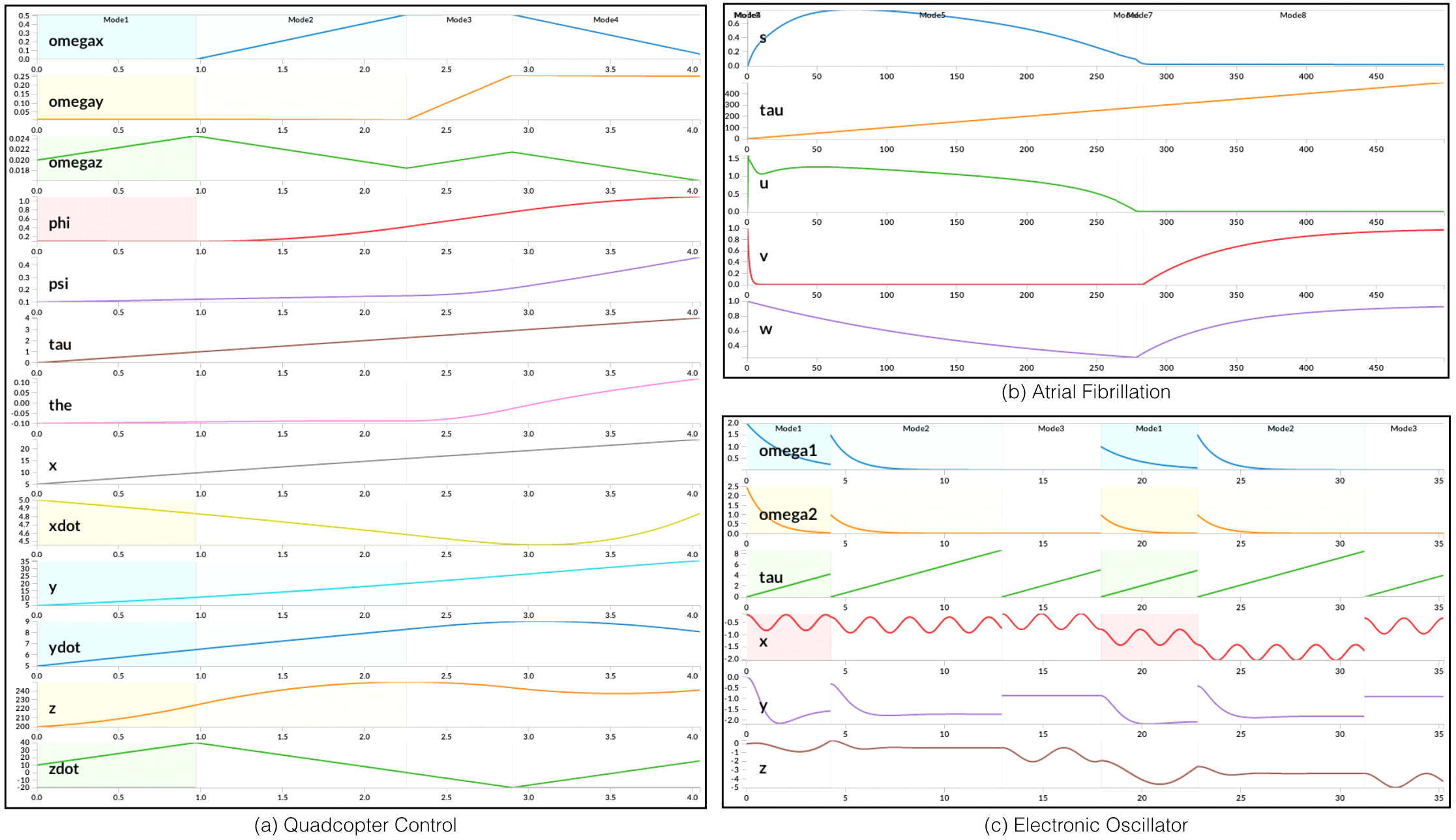}
\caption{Example trajectories computed for the following models: (a) Quadcopter Control, (b) Atrial Fibrillation, (c) Electronic Oscillator.}\label{fig:exp_figure}
\end{center}
\end{figure}
\section{Discussion}

Reachable set computation, which computes geometric representations of the complete set of reachable states, is the mainstream approach for analyzing bounded reachability of hybrid systems. The techniques can have difficulty in scaling on systems with very complex dynamics and discrete transitions. Bounded model checking has the advantage of focusing the search for one counterexample, and does not maintain the complete set of reachable states. With fast SAT/SMT solvers, bounded model checking techniques can natively handle the discrete components in hybrid systems. Bounded model checking requires a very powerful solver, one that can handle ODEs and nested quantifiers. We have proved that the complexity of bounded $\delta$-reachability is comparable to SAT solving, and it is reasonable to expect that with more improvement on the solver, large realistic systems can eventually be handled in practice. Note again that all the techniques in reachable set computation can be directly used in logic solvers, and it is possible to have practical tools that combine the advantages of both approaches.
\section{Conclusion}
We developed the framework of $\delta$-complete analysis for bounded reachability of a wide range of hybrid systems. $\delta$-Complete reachability analysis reduces verification problems to $\delta$-decision problems of formulas over the reals. It follows from $\delta$-decidability of these formulas that $\delta$-complete reachability analysis of a wide range of nonlinear hybrid systems is decidable. In practice, $\delta$-reachability problems are solved through reduction to $\delta$-decision problems for first-order formulas over the reals. We demonstrated the scalability of our approach on highly nonlinear hybrid systems.
\bibliographystyle{abbrv}
\bibliography{tau}
\newpage
\section*{Appendix}
\paragraph{Proof of Lemma~\ref{equiv-delta}. }
\begin{proof}
We prove for the case with nontrivial invariants. We work with the unperturbed encoding, which easily applies to the $\delta$-perturbed version. We will need to do induction on the subformula of $\reach_{H,U}$ that does not contain the unsafe conditions. For reasons that will be made clear below, we split the formula $\reach_{H,U(k,M)}$ into two parts and write it as the conjunction $\traj(k,M)\wedge \unsafe(k)$, where $\unsafe(k)$ is $\bigvee_{q\in Q}(\unsafe_q(\vec x_k^t)\wedge \enforce(q,k))$.

Suppose $\mathbb{R}\models \reach_{H,U}(k,M)$. We do induction on $k$ to prove that there exists a trajectory $\xi\in\llbracket H\rrbracket$ that contains $k$ mode changes. When $k=0$, without loss of generality we pick an arbitrary starting mode $q$, such that the $\traj(k,M)$ part of the formula can be simplified as
\begin{eqnarray*}
& &\exists^X \vec x_{0} \exists^X\vec x_{0}^t\exists^{[0,M]}t_0\;\Big(\init_{q}(\vec x_{0})\wedge \flow_{q}(\vec x_{0}, \vec x_{0}^t, t_0)\wedge \enforce(q,0)\\
& &\hspace{5cm}\wedge \forall^{[0,t_0]}t\forall^X\vec x\;(\flow_{q}(\vec x_{0}, \vec x, t)\rightarrow \inv_{q}(\vec x))).
\end{eqnarray*}
Since the formula is true, there exists witnesses $\vec a, \vec a^t, \tau$ such that the quantifier-free part is satisfied. By well-definedness of $\flow_q$ there exists a trajectory $\xi$ from $\vec a_0$ to $\vec a^t$ such that for any $0\leq \tau'\leq \tau$, $\xi(\tau)$ satisfies the invariant condition. Now, suppose $k= (k-1)+1$ ($k\geq 1$) and by inductive hypothesis there exists a trajectory $\xi'\in\llbracket H\rrbracket$ with $k-1$ mode changes. We now extend $\xi'$ with one more mode change. Let $\traj(k-1,M)$ be the part of $\reach_{H,U})(k-1,M)$, and thus $\traj(k,M)$ can be written as
\begin{eqnarray*}
& &\exists \vec x_{k}\exists^X\vec x_{k}^t \exists^{[0,M]}t_k\;\\
& &\Bigg(\traj(k-1,M)\wedge \bigvee_{q, q'\in Q} \Big(\jump_{q\rightarrow q'}(\vec x_{k-1}^t, \vec x_{k})\wedge \flow_{q'}(\vec x_{k}, \vec x_{k}^t, t_{k})\\
& &\wedge \enforce(q,q',i)\wedge \forall^{[0,t_{k}]}t\forall^X\vec x\;(\flow_{q'}(\vec x_{k}, \vec x,t)\rightarrow \inv_{q'}(\vec x)) )\wedge\enforce(q',k)\Big)\Bigg)
\end{eqnarray*}
Note that $\vec x_0,...,\vec x_{k-1}^t$ are quantified variables in $\traj(k-1,M)$. Since the formula is true, there exists $\vec a_{k}, \vec a_k^t, \tau_k$ that witness the satisfiability of the quantifier-free part of the formula outside of $\traj(k-1,M)$. Now, we extend $\xi'\in\llbracket H\rrbracket$ in the following way. Let the last state of $\xi'$ be given by $\vec a^t_{k-1}$. Following the formula, we have that $\jump_{1\rightarrow q'}(\vec a_{k-1}^t, \vec a_k)$ satisfies the jumping condition between mode $q$ and $q'$. It is then followed by a continuous trajectory that starts from $\vec a_k$ and ends at $\vec a_k^t$, satisfying $\flow(\vec a_k, \vec a_k^t, \tau_k)$. Thus, there exists a trajectory $\xi\in \llbracket H\rrbracket$ with $k$ mode changes. Thus, for all $k$ there exists a trajectory $\xi\in \llbracket H\rrbracket$ such that for some $(k,t)\in T(\xi)$, $\xi(k,t),\sigma_{\xi}(k)\in \llbracket \unsafe\rrbracket$.

The reverse direction is easy. Suppose there exists a trajectory $\xi\in \llbracket H\rrbracket$ such that for some $(k,t)\in T(\xi)$, $\xi(k,t),\sigma_{\xi}(k)\in \llbracket \unsafe\rrbracket$, then the start and end points in each piece of the continuous trajectories witness the formula $\reach_{H,U}(k,M)$.
\qed

\end{proof}

\end{document}